\pdfoutput=1
\documentclass[10pt]{article}

\usepackage[section]{placeins}

\usepackage[linesnumbered, ruled,vlined]{algorithm2e}
\usepackage{tikz}

\usepackage{bibspacing}
\usepackage[paper=letterpaper,centering]{geometry}
\usepackage{graphics}
\usepackage{lhelp}
\usepackage{bbm}
\usepackage{amsmath,amscd}
\usepackage{graphicx,pifont,color} 
\usepackage{url}
\usepackage{pstricks,pst-node,pst-grad,pst-plot}
\usepackage{amsfonts}
\usepackage{epsfig, color}
\usepackage{latexsym}
 \usepackage{amsthm}
\usepackage{multirow}
\usepackage{comment}
\usepackage{longtable}

\def\C {\ensuremath{\mathbb{C}}}

\def\K {\ensuremath{\mathbf{k}}}

\def\Q {\ensuremath{\mathbb{Q}}}
\def\R {\ensuremath{\mathbb{R}}}

\def\KK {\ensuremath{\mathbf{K}}}

\newcommand{\xx}{\mathbf{x}}

\def\x {\ensuremath{\mathbf{x}}}

\newtheorem{Theorem}{Theorem}

\newtheorem{Example}{Example}

\newcommand{\lc}[1]{\mbox{{\rm lc}$(#1)$}}

\newcommand{\init}[1]{\mbox{{\rm init}$(#1)$}}

\newcommand{\mdeg}[1]{\mbox{{\rm mdeg}$(#1)$}}
\newcommand{\mvar}[1]{\mbox{{\rm mvar}$(#1)$}}
\newcommand{\prem}[1]{\mbox{{\rm prem}$(#1)$}}
\newcommand{\pquo}[1]{\mbox{{\rm pquo}$(#1)$}}

\newcommand{\res}[1]{\mbox{{\rm res}$(#1)$}}

\newcommand{\tail}[1]{\mbox{{\rm tail}$(#1)$}}
\newcommand{\coeff}[1]{\mbox{{\rm coeff}$(#1)$}}
\renewcommand{\min}[1]{\mbox{{\rm min}$(#1)$}}

\newcommand{\RegularizeInitial}[1]{\mbox{{\sf MakeLeadingCoefficientInvertible$_n$}$(#1)$}}

\newcommand{\RegularGcd}[1]{\mbox{{\sf Gcd$_n$}$(#1)$}}

\newcommand{\Squarefree}[1]{\mbox{{\sf Squarefree$_n$}$(#1)$}}

\newcommand{\SquarefreePart}[1]{\mbox{{\sf SquarefreePart}$(#1)$}}

\newcommand{\Maple}{{\sc  Maple}}

\newcommand{\RegularChains}{{\tt  Regu\-lar\-Chains}}

\def\x {\ensuremath{\mathbf{x}}}

\def\DD {\ensuremath{\mathcal{D}}}

\def\TT {\ensuremath{T}}

\newcommand{\der}[1]{\mbox{{\rm der}$(#1)$}}
\newcommand{\discrim}[1]{\mbox{{\rm discrim}$(#1)$}}

\newcommand{\level}[1]{\mbox{{\rm level}$(#1)$}}

\newcommand{\Project}[1]{\mbox{{\sf ExtractProjection}$(#1)$}}

\newcommand{\CylindricalDecompose}[1]{\mbox{{\sf CylindricalDecompose}$(#1)$}}
\newcommand{\CoFactor}[1]{\mbox{{\sf CoFactor}$(#1)$}}

\newcommand{\NextPathToDo}[1]{\mbox{{\sf NextPathToDo$_{n-1}$}$(#1)$}}

\newcommand{\Update}[1]{\mbox{{\sf UpdatePath}$(#1)$}}

\newif\ifcomment
\commentfalse

\newcommand{\CCAD}{{PCAD}}
\newcommand{\TCAD}{{\sc TCAD}}
\newcommand{\IntersectPath}[1]{\mbox{{\sf IntersectPath}$(#1)$}}
\newcommand{\CAD}{{CAD}}
\newcommand{\QE}{{QE}}
\newcommand{\QEPCAD}{{\sc Qepcad}}

\begin{document}

\begin{center}
  {\Large\bf 
    An Incremental Algorithm for Computing Cylindrical Algebraic Decompositions
  }
\mbox{}\\[11pt]
{\large
  Changbo Chen, Marc Moreno Maza
}
\mbox{}\\[5pt] 
ORCCA, University of Western Ontario (UWO) \\
London, Ontario, Canada \\
{\tt \{cchen252,moreno\}@csd.uwo.ca}\\[5pt]
\end{center} 

\begin{abstract}
In this paper, we propose 
an incremental algorithm for 
computing cylindrical algebraic decompositions.
The algorithm consists of two parts: 
computing a complex cylindrical tree 
and refining this complex tree into a 
cylindrical tree in real space. 
The incrementality comes from 
the first part of the algorithm, 
where a complex cylindrical tree
is constructed by refining a previous complex cylindrical tree
with a polynomial constraint. 
We have implemented our algorithm in Maple. 
The experimentation shows that 
the proposed  algorithm outperforms existing ones 
for many examples taken from the literature.
\end{abstract}

\section{Introduction}

Cylindrical algebraic decomposition (CAD)
is a fundamental tool in real algebraic geometry. 
It was invented by G.E. Collins in 1973~\cite{col75}
for solving real quantifier elimination (QE) problems.
In the last forty years, following Collins' original 
projection-lifting scheme, 
many enhancements have been performed in order to 
ameliorate the efficiency of CAD construction, 
including adjacency and clustering techniques~\cite{Arnon84b},
improved projection methods~\cite{scott88,hong90,CavinesJohnson98,brown01}, 
partially built {\CAD}s~\cite{ch91,scott93,adam00},
improved stack construction~\cite{Collins02},
efficient projection orders~\cite{Dolzmann04}, 
making use of equational 
constraints~\cite{Collins98, McCallum2001, Brown05, McCallum2009}, 
and so on.
Moreover, CADs can be computed 
by several software packages, such as 
{\sc Qepcad}~\cite{QEPCAD, Bro03}, {\sf Mathematica}~\cite{adam00, adam06}, 
{\sf Redlog}~\cite{Dolzmann96} and {\sf SyNRAC}~\cite{Iwane09}.

In~\cite{CMXY09}, together with B. Xia and L. Yang,
we presented a different way for computing CADs 
based on triangular decomposition of polynomial systems.
In that paper, we introduced the concept of 
cylindrical decomposition of the complex space (CCD),  
from which a CAD can be easily derived.
The concept of CCD is reviewed in Section~\ref{sec:complex}.
In the rest of the present paper, we use {\TCAD}
to denote CAD based on triangular decompositions
while {\CCAD} refers to CAD based on Collins' projection-lifting scheme.

The CCD part of {\TCAD} can be seen as an enhanced 
projection phase of {\CCAD}. 
However, w.r.t. {\CCAD} (especially when the projection operator 
is using Collins'~\cite{col75} or Hong's~\cite{hong90}),
the ``case discussion'' scheme of {\TCAD} avoids unnecessary computations 
that projection operator performs on unrelated branches. 
In addition, one observes that 
the reason why McCallum's~\cite{scott98} 
(including Brown's~\cite{brown01}) projection operators
may fail for some examples is due to the fact that 
they are missing a ``case discussion'' scheme.
McCallum's operator relies on the assumption that generically all coefficients
of a polynomial\footnote{More precisely, a multivariate polynomial
regarded as a univariate one with respect to its main variable.}
will not vanish simultaneously above a positive-dimensional component. 
If this assumption fails, then this operator is replaced by Collins-Hong
projection-operator~\cite{hong90}.
The fact that all coefficients of polynomial could vanish 
simultaneously above some component is never a problem in {\TCAD}.
For this reason, we view it as  an improvement of previous works.

Trying to use sophisticated algebraic elimination techniques to
improve CAD constructions is not a new idea. 
In papers~\cite{BH91, WBD12}, the authors investigated how to 
use Gr{\"o}bner bases to preprocess the input system 
in order to make the subsequent CAD computations more efficient.
The main difference between these two works and the work of~\cite{CMXY09}
is that the former approach is about preprocessing input for CAD while the latter
one  presents a different way of constructing CADs.

In~\cite{CMXY09}, the focus was on how to apply 
triangular decomposition techniques to compute CADs.
To this end, lots 
of existing high-level routines were used to facilitate
explaining ideas. 
These high-level routines involve many black-boxes, 
which hide many unnecessary or redundant computations. 
As a result, the computation time of {\TCAD} is much higher than that of {\CCAD},
although {\TCAD} computes usually less cells~\cite{Chen11}.

In the present paper, 
we abandon those black-boxes and compute {\TCAD} from scratch. 
It turns out that the key solution for avoiding redundant computations 
is to compute CCD in an {\em incremental manner}.
The same motivation and a similar strategy
appeared in~\cite{moreno00, CM11} in the context of
triangular decomposition of algebraic sets.
The core operation of such an incremental algorithm 
is an {\sf Intersect} operation, 
which refines an existing cylindrical tree 
w.r.t. a polynomial.
We dedicate Section~\ref{sec:incremental} to presenting a complete 
incremental algorithm for computing {\TCAD} 
by means of this {\sf Intersect} operation.

In~\cite{Strzebonski2010a}, the author
presented an algorithm for computing 
with semi-algebraic sets represented by cylindrical algebraic formulas. 
That algorithm also allows computing CAD in an incremental manner.
The underlying technique is based on the projection-lifting 
scheme where one first computes projection factor sets by 
a global projection operator. 
In contrast, 
the incremental algorithm presented here, 
is conducted by refining different branches 
of an existing tree via GCD computations.

This {\sf Intersect} operation can systematically take advantage of equational constraints. 
The problem of making use of equational constraints in CAD has been 
studied by many researchers~\cite{Collins98, McCallum2001, Brown05, McCallum2009}.
In Section~\ref{sec:equation},
we provide a detailed discussion on how we solve 
this problem.

When  applied to a polynomial system 
having finitely many complex solutions, our incremental CCD algorithm
specializes into computing a triangular decomposition,
say ${\cal D}$, 
such that the zero sets of the output regular chains are
disjoint.
Moreover, such a decomposition 
has no critical pairs in the sense of the 
equiprojectable decomposition algorithm of~\cite{DMSWX05a}.
This implies that only the ``Merge'' part of the 
``Split \& Merge'' algorithm of~\cite{DMSWX05a}
is required for turning ${\cal D}$ into an
 equiprojectable decomposition (which is a canonical
representation of the input variety, once the variable order is fixed).
Consequently, one could hope extending the
notion of equiprojectable decomposition (and related algorithms) 
to positive dimension by means of our incremental CCD algorithm.
This perspective can be seen as an indirect application 
of CAD to triangular decomposition.

As we shall review in Section~\ref{sec:complex}, 
a CCD is encoded by a tree data-structure.
 Then each path of this tree is a simple system
in the sense of~\cite{Thomas37, Wang98a}.
So the work presented here can also be used to 
compute a Thomas decomposition of a polynomial system~\cite{Wang98a, Thomas10}. 
Moreover, the decomposition we compute is not only disjoint, 
but also cylindrically arranged. 

The complexity of our algorithm cannot be better 
than doubly exponential in the number of variables~\cite{BrownDavenport2007}.
So the motivation of our work is to suggest possible ways 
to improve the practical applicability of CAD.
The benchmark in Section~\ref{sec:benchmark} shows that {\TCAD}
outperforms {\QEPCAD}~\cite{QEPCAD,Bro03} and {\sf Mathematica}~\cite{adam00}
for many well-known examples. 
The algorithm presented in this paper can support {\QE}. 
We have realized a preliminary implementation of an algorithm for doing {\QE}
via {\TCAD}. We will report on this work in a future paper.

\section{Complex cylindrical tree}
\label{sec:complex}

Throughout this paper, we consider a field $\K$ of characteristic zero and
denote by $\KK$ the algebraic closure of $\K$.
Let $\K[\xx]$ be the polynomial ring over the
field $\K$ with ordered variables $\xx= x_1 < \cdots < x_n$. Let $p\in {\K}[\xx]$ be
a non-constant polynomial and $x \in \xx$ be a variable.
We denote by ${\rm deg}(p,x)$ and ${\rm lc}(p,x)$
the degree and the leading coefficient of $p$ w.r.t. $x$.
The greatest variable appearing in $p$ is
called the {\em main variable}, denoted by $\mvar{p}$. 
The leading coefficient, the degree, the reductum of  $p$ w.r.t. $\mvar{p}$ are
called the {\em initial}, the {\em main degree}, the {\em tail} of $p$; 
they are denoted by $\init{p}$, $\mdeg{p}$, $\tail{p}$ respectively. 
The integer $k$ such that $x_k=\mvar{p}$ is called 
the {\em level} of the polynomial $p$.
We denote by $\der{p}$ the derivative  of $p$ w.r.t. $\mvar{p}$.
The notions presented below were introduced
in~\cite{CMXY09} and they are illustrated at the beginning
of Section~\ref{sec:datastructure}.

\smallskip\noindent{\small \bf Separation.}
Let $C$ be a subset of $\KK^{n-1}$
and ${P}\subset\K[x_1,\ldots,x_{n-1}, x_n]$ be a finite set of level $n$ 
polynomials.
We say that $P$ {\em separates above} $C$ if for each $\alpha\in C$:
\begin{itemizeshort}
\item for each $p\in{P}$, the polynomial \init{p} 
      does not vanish at $\alpha$, 
\item the polynomials $p(\alpha,x_n) \in \KK[x_n]$, for all $p\in{P}$, 
              are squarefree and coprime.
\end{itemizeshort}
Note that this definition allows $C$ to be a semi-algebraic set, 
see Theorem~\ref{Theorem:stack}.

\smallskip\noindent{\small \bf Cylindrical decomposition.}
By induction on $n$,
we define the notion of a {\em cylindrical decomposition of} ${\KK}^n$
together with that of the {\em tree associated with a 
cylindrical decomposition of} ${\KK}^n$.
For $n=1$, a cylindrical decomposition of $\KK$ is a 
finite collection of sets ${\DD} = \{D_1,\ldots,D_{r+1}\}$, 
where either $r=0$ and $D_1=\KK$, or $r>0$ and
there exists $r$  non-constant coprime squarefree polynomials 
$p_1,\ldots,p_{r}$ of $\K[x_1]$ such that for $1\leq i\leq r$
we have 
$
D_i=\{x_1\in\KK\mid p_i(x_1)=0\}, 
$
and 
$
D_{r+1}=\{x_1\in\KK\mid p_1(x_1)\cdots p_r(x_1) \neq 0\}.
$
Note that the $D_i$'s, for all $1\leq i\leq {r+1}$, form a partition
of $\KK$. 
The tree associated  with ${\DD}$ is a rooted tree 
whose nodes, other than the root, are $D_1, \ldots, D_r, D_{r+1}$
which all are leaves and children of the root.
Now let $n>1$, and let ${\DD}'=\{D_1,\ldots,D_s\}$ be any cylindrical 
decomposition of $\KK^{n-1}$. For each $D_i$, 
let $r_i$ be a non-negative integer and 
let $\{p_{i,1},\ldots,p_{i,r_i}\}$ be a set of polynomials 
which separates above $D_i$. 
If $r_i=0$, set $D_{i,1}=D_i\times\KK$. 
If $r_i>0$, set
$$
D_{i,j}=\{(\alpha,x_n)\in\KK^n\mid \alpha\in D_i 
\ {\rm and} \  p_{i,j}(\alpha,x_n)=0\},
$$
for $1\leq j\leq r_i$ and set
$
D_{i,r_{i}+1}=\left\{(\alpha,x_n)\in\KK^n\mid \alpha\in D_i
\ {\rm and} \  \left(\prod_{j=1}^{r_i}p_{i,j}(\alpha,x_n)\right)\neq0\right\}.
$
The collection 
${\DD}=\{D_{i,j}\mid 1\leq i\leq s, 1\leq j\leq r_{i}+1\}$ is called a
{\em cylindrical decomposition} of $\KK^n$. 
The sets $D_{i,j}$ are called the {\em cells} of ${\DD}$.
If ${\TT}'$ is the tree associated with ${\DD}'$
then the tree ${\TT}$ associated with  ${\DD}$ is defined 
as follows.
For each $1\leq i\leq s$, the set $D_i$ is a leaf
in ${\TT}'$ which has all $D_{i,j}$'s for children 
in ${\TT}$; thus the $D_{i,j}$'s are the leaves of ${\TT}$.

Note that each node $N$ of ${\TT}$ is 
either associated with no constraints,
or associated with a polynomial
constraint, which itself is either an equation or an inequation.
Note also that, if the level of the polynomial defining the constraint
at $N$ is ${\ell}$, then ${\ell}$ is the length of a path
from $N$ to the root.
Moreover, the polynomial constraints
along a path from the root to a leaf form 
a polynomial system called a {\em cylindrical system}  
{\em of}  $\K[x_1,\ldots,x_n]$ {\em induced by} ${\TT}$.
Let $S$ be such a cylindrical system.
We denote by $Z(S)$ the zero set of $S$.
Therefore, each cell of ${\DD}$ is the zero set of a 
cylindrical system induced by ${\TT}$.

Let $\Gamma$ be a sub-tree of ${\TT}$ such that
the root of $\Gamma$ is that of ${\TT}$.
Then, we call \, $\Gamma$ a {\em cylindrical tree of} $\K[x_1,\ldots,x_n]$
{\em induced by} ${\TT}$.
This cylindrical tree $\Gamma$ is said {\em partial} if it admits
a non-leaf node $N$ such that
the zero set of the constraint of $N$ is not equal
to the union of the zero sets of the constraints
of the children of $N$.
If $\Gamma$ is not partial, then it is called {\em complete}.

In the algorithms of Section~\ref{sec:incremental},
the cylindrical tree is an essential data structure.
Section~\ref{sec:datastructure} discusses the
main properties and operations on this data structure.

Let $F=\{f_1,\ldots,f_s\}$ be a finite set of polynomials 
of $\K[x_1<\cdots<x_n]$. 
A cylindrical decomposition ${\DD}$ of $\KK^n$ is 
called {\em $F$-invariant}
if for any given cell $D$ of  ${\DD}$ and any given polynomial $f\in F$, 
either $f$ vanishes at all points of $D$ or $f$ vanishes at no points of $D$.

\begin{Example}
Let $F := \{y^2+x, y^2+y\}$. 
An $F$-invariant cylindrical decomposition of $\C^2$
is illustrated by Figure~\ref{fig:tree}.
\begin{center}
\begin{figure}[htbp]
\scalebox{0.4}{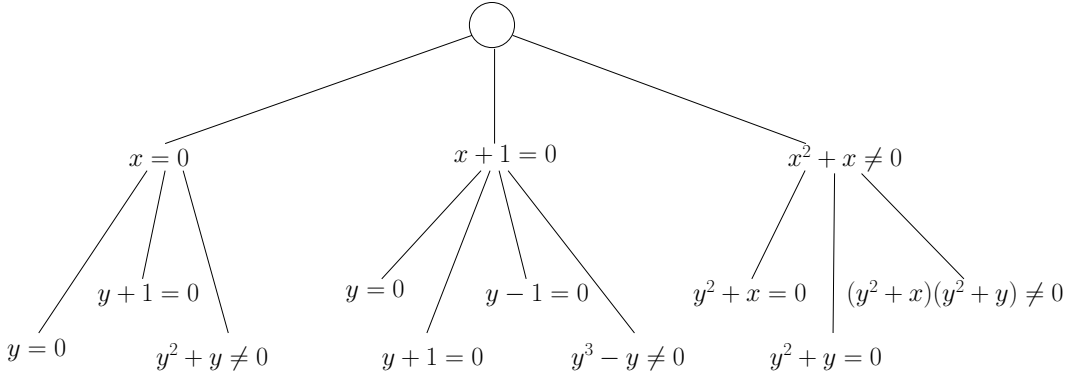 }
\caption{an $F := \{y^2+x, y^2+y\}$ invariant complex cylindrical tree}
\label{fig:tree}
\end{figure}
\end{center}
\end{Example}

We observe that every cylindrical system induced by a cylindrical tree
is 
a {\em simple system}, as defined by Wang in~\cite{Wang98a}.
This notion was first introduced by Thomas in 1937~\cite{Thomas37}.
Simple systems have many nice properties. For example, 
if $[A,B]$ is a simple system, then 
the pair $[A,\prod_{p\in B} p]$ is a squarefree 
regular system, as defined by Wang in~\cite{Wang98a,Wang00}.

Let $\Gamma$ be a cylindrical system of $\K[\x]$ and
let $p$ be a polynomial of $\K[\x]$.
We say that $p$ is {\em invertible modulo} $\Gamma$ if 
for any $\alpha \in Z(\Gamma)$, we have $p(\alpha) \neq 0$.
We say that $p$ is {\em zero modulo} $\Gamma$ 
if for any $\alpha \in Z(\Gamma)$, we have $p(\alpha)=0$. 
We say that $p$ is {\em sign invariant above} $\Gamma$ if $p$
is either zero or invertible modulo $\Gamma$.
Let $q$ be another polynomial of $\K[\x]$.
We say that $p=q$ modulo $\Gamma$ 
if $Z(\Gamma)\cap Z(p)=Z(\Gamma)\cap Z(q)$.

\smallskip\noindent{\small \bf Greatest common divisor (GCD).}
Let $p$ and $f$ be two level $n$ polynomials in $\K[\x]$.
Let $\Gamma$ be a cylindrical system of $\K[x_1,\ldots,x_{n-1}]$.
For any $u\in\KK^{n-1}$ of $Z(\Gamma)$, assume at least one of $\lc{p,x_n}(u)$
and $\lc{f,x_n}(u)$ is not zero.
A polynomial $g \in \K[\x]$ is called a {\em GCD of}
$p$ {\em and} $f$ {\em modulo} $\Gamma$
if for any $u\in\KK^{n-1}$ of $Z(\Gamma)$,
\begin{itemizeshort}
\item $g(u)$ is a GCD of $p(u)$ and $f(u)$ in $\KK[x_n]$, and
\item we have $\lc{g,x_n}(u)\neq 0$.
\end{itemizeshort}
Let $d_p=\deg(p, x_n)$, $d_f=\deg(f, x_n)$.
Recall that we assume $d_p,d_f\geq 1$.
Let $\lambda=\min{d_p, d_f}$.
Let $\Gamma$ be a cylindrical system of $\K[x_1,\ldots,x_{n-1}]$.
Let $S_0,\ldots,S_{\lambda-1}$ be the subresultant polynomials~\cite{Mis93, Ducos00} 
of $p$ and $f$ w.r.t. $x_n$.
Let $s_i=\coeff{S_i, x_n^{i}}$ be the principle
subresultant coefficient of $S_i$, for $0 \leq i \leq \lambda-1 $.
If $d_p\geq d_f$, we define $S_{\lambda}=f$, $S_{\lambda+1}=p$, 
$s_{\lambda}=\init{f}$ and $s_{\lambda+1}=\init{p}$.
If $d_p<d_f$, we define $S_{\lambda}=p$, $S_{\lambda+1}=f$, 
$s_{\lambda}=\init{p}$ and $s_{\lambda+1}=\init{f}$.

\begin{Theorem}
\label{Theorem:gcd}
Let $j$  be an integer, with $1\leq  j \leq \lambda+1$,
such that $s_j$ is invertible modulo $\Gamma$
and such that for any $0\leq i< j$, we have $s_i=0$ modulo $\Gamma$.
Then $S_j$ is a GCD of $p$ and $f$ modulo $\Gamma$.
\end{Theorem}
\begin{proof}
It can be easily proved by the specialization property 
of subresultant chains. 
In particular, it is a direct corollary of 
Theorem 5 in~\cite{CM12}.
\end{proof}

\section{Data structure for cylindrical decomposition}
\label{sec:datastructure}

In this section, we describe the data-structures
that are used by the algorithms presented in this paper
for computing cylindrical decompositions.
To understand the motivation of our algorithm design, let us
consider a simple example with $n = 2$ variables.
Let $a, b$ 
be two coprime squarefree non-constant univariate polynomials in $k[x_1]$.
Observe that $L := k[x_1] / \langle a \, b \rangle$ is a direct product
of fields.
Let also $c, d$ be two bivariate polynomials 
of $k[x_1, x_2]$, such that
 ${\rm deg}(c,x_2) > 0$, ${\rm deg}(d,x_2) > 0$, 
and ${\rm lc}(c,x_2) =  {\rm lc}(d,x_2) = 1$ hold
and such that $c, d$ are coprime and squarefree univariate
as polynomials of 
$L[x_2]$.
Therefore the following four polynomial systems
are simple systems
\begin{equation*}
\begin{array}{c}
\left\{
\begin{array}{c}
a(x_1) b(x_1) = 0 \\
c(x_1, x_2) = 0
\end{array}
\right., \  
\left\{
\begin{array}{c}
a(x_1) b(x_1) = 0 \\
d(x_1, x_2) = 0
\end{array}
\right., \
\left\{
\begin{array}{c}
a(x_1) b(x_1) = 0 \\
c(x_1, x_2) d(x_1, x_2) \neq 0
\end{array}
\right., \  
\left\{
\begin{array}{c}
a(x_1) b(x_1) \neq 0 \\
\end{array}
\right.
\end{array}
\end{equation*}
that we denote respectively by $S_1, S_2, S_3, S_4$.
It is easy to check that the 
zero sets $Z(S_1)$, $Z(S_2)$, $Z(S_3)$, $Z(S_4)$
are the cells of a cylindrical decomposition ${\DD}$ of ${\KK}^2$.

Let $f \in k[x_1]$ be another univariate polynomial.
Assume that one has to {\em refine} ${\DD}$ into a cylindrical decomposition
of ${\KK}^2$ which is required to be $\{ f \}$-invariant.
That is, one has to test whether $f$ is invertible or zero
modulo each of the systems $S_1, S_2, S_3, S_4$,
and further decompose when appropriate.
Assume that the polynomial 
$a$ divides $f$ whereas $b, f$ are coprime.
Assume also that the system $S_1$ is processed first in time.
By computing ${\rm gcd}(f,ab)$, which yields $a$, 
one splits $S_1$ into the following two
sub-systems that we denote by $S_{1,1}$ and $S_{1,2}$.
\begin{equation*}
\left\{
\begin{array}{c}
a(x_1) = 0 \\
c(x_1, x_2) = 0
\end{array}
\right., \  {\rm and} \ 
\left\{
\begin{array}{c}
b(x_1) = 0 \\
c(x_1, x_2) = 0.
\end{array}
\right.
\end{equation*}
Assume that $S_2$ is processed next.
By computing ${\rm gcd}(f,ab)$ (again) 
one splits $S_2$ into the following two
sub-systems that we denote by $S_{2,1}$ and $S_{2,2}$.
\begin{equation*}
\left\{
\begin{array}{c}
a(x_1) = 0 \\
d(x_1, x_2) = 0
\end{array}
\right., \  {\rm and} \ 
\left\{
\begin{array}{c}
b(x_1) = 0 \\
d(x_1, x_2) = 0.
\end{array}
\right.
\end{equation*}
Consequently, in the course of the creation 
of $S_{1,1}, S_{1,2}$, $S_{2,1}$ and $S_{2,2}$,
the same polynomial GCD  and the same 
field extensions (namely $k[x_1] / \langle a \rangle$
and $k[x_1] / \langle b \rangle$) were computed twice.
This duplication of calculation and data is a common
phenomenon and a performance bottleneck in most
algorithms for decomposing polynomial systems.

Mathematically, each constructible set should not be
represented more than once in a computer program.
To implement this idea, all constructible sets
manipulated during the execution of a given computer 
program should be seen as part of the same {\em universe},
say ${\KK}^n$. Moreover, the subroutines of this program
should have the same view on the universe, which is then
a {\em shared data-structure}, such that whenever a
subroutine modifies the universe all subroutines
have immediate access to the modified universe.
Satisfying these requirements is a well-known 
challenge in computer science, an instance of which
is the question of {\em memory consistency}
for shared-memory parallel computer architectures,
such as multicores. With our above example, even if 
we do not intend to run computations concurrently,
we are concerned with the practical efficiency
and ease-of-use of the mechanisms that maintain {\em up-to-date}
all views on the universe.

Recall that a cylindrical decomposition can be identified to
a  tree 
where each node is a constructible set of ${\KK}^n$ given
by either an equation constraint, or an inequation
constraint, or no constraints at all.
In this latter case, the corresponding constructible
set is the whole space.
All algorithms in Section~\ref{sec:incremental}
work on a given cylindrical decomposition ${\DD}$
encoded by a tree ${\TT}$ (as defined in 
Section~\ref{sec:complex}).
That is, the tree $T$ is regarded as the universe.

We assume that there is a procedure for updating the 
tree ${\TT}$, which, given a ``node-to-be-replaced'' $N$ and 
its ``replacing nodes'' $N_1, \ldots, N_e$, 
is called {\sf split}$(N; N_1, \ldots, N_e)$ and 
works as follows:
\begin{enumerateshort}
\item for $i = 1, \ldots, e$,  for each child $C$ of $N$ deeply copy (thus 
      creating new nodes) the sub-tree rooted at $C$
      and make that copy of $C$ a child of $N_i$,
\item update the parent of $N$ such that $N_1, \ldots, N_e$
      are new children of the the parent of $N$,
\item remove the entire sub-tree rooted at $N$ from the universe, 
      including $N$.
\end{enumerateshort}
We assume that all updates are performed sequentially (thus using
mutual exclusion mechanism in case of concurrent execution
of the algorithms of Section~\ref{sec:incremental})
such that no data-races can occur.

We also assume that each node $N$  (whether it is a node 
in the present or has been removed from the universe)
has a unique key, called 
{\sf key}$(N)$, and a data field, called {\sf value}$(N)$,
storing various information including:
\begin{itemizeshort}
\item a time stamp {\sc past} or {\sc present},
\item if {\sc past}, the list of its replacing nodes (as specified
      with the {\sf split} procedure) and the list
      of its children at the time it was replaced,
\item if {\sc present}, the list of its  children and a pointer to the parent.
\end{itemizeshort}

All nodes are stored in a {\em dictionary} $H$
which can be accessed by all subroutines.
Modifying the universe means updating $H$
using the {\sf split} procedure.
Since all our algorithms stated in Section~\ref{sec:incremental} are sequential,
no synchronization issue has to be considered.
The mechanism described above allows us
to achieve our goals.

\section{Constructing a cylindrical tree incrementally}
\label{sec:incremental}

In this section, we present an incremental algorithm
for computing a cylindrical tree, as defined in Section~\ref{sec:complex}.
We start by commenting on the style of the pseudo-code.
Secondly, we present the specifications of the algorithm
and related subroutines. 
Thirdly, we state all the algorithms in pseudo-code style. 
Finally, proof sketches of the algorithms are provided at the end of
this section.

Following the principles introduced in Section~\ref{sec:datastructure},
our procedures operate on a ``universe'' (which is a
cylindrical tree $T$) that they modify when needed.
These modifications are of two types:
\begin{itemizeshort}
\item splitting a node,
\item attaching information to a node.
\end{itemizeshort}
In addition to the attributes described in Section~\ref{sec:datastructure},
a node has attributes corresponding to the results of operations
like {\sf Squarefree}, {\sf Gcd}, {\sf Intersect}.
In other words, our procedures
do not return values; instead they store their
results in the nodes of the universe.
This technique greatly simplifies pseudo-code.

Since attributes of nodes are intensively used
in our pseudo-code, we use the standard ``dot''
notation of object oriented programming languages.
In addition, since a node can have many attributes,
we make the following convention.
Suppose that a node $V$ is split into two nodes
$V_1$ and $V_2$.
Some attributes are likely
to have  different values in  $V_1$ (resp. $V_2$)
and $V$. But most of them will often have
the same values in both nodes.
Therefore, after setting up the values of the
attributes that differ, we simply write
$V_1$.others := $V$.others
to define the attributes of $V_1$ whose values are unchanged
w.r.t. $V$.

Several procedures iterate through all the paths of the universe $T$.
By path, we mean a path (in the sense of graph theory)
from the root of $T$ to a leaf of $T$.
The current path is often denoted by $\Gamma$ or $C$.
Recall from Section~\ref{sec:complex} 
that a path in $T$ corresponds to a simple system, say $S$.
Computing modulo $S$ may split $S$ and thus modify the universe {\em automatically},
that is, in a transparent manner in the pseudo-code.
However, splitting $S$ also changes the current path.
For clarity, we explicitly invoke a function called {\sf UpdatePath},
which updates its first argument (namely the current path) from the universe.

In order to iterate through all the paths of the universe $T$,
we use a function {\sf NextPathToDo}.
This command is a {\em generator} or an {\em iterator}
in the sense of the theory of programming languages.
That is, it views $T$ as a stream of paths and returns
the next path-to-be-visited, if any.
Thanks to the fact that the universe is always up-to-date,
the function {\sf NextPathToDo} is able to return the next path-to-be-visited
in the current state of the universe.

A frequently used operation on the universe and its paths
is {\sf ExtractProjection}, see for instance Algorithm~\ref{Algo:IntersectPath}.
When applied to the universe $T$ and an integer $k$ (for $0 \leq k < n$,
where $n$ is the length of a path from the root of $T$ to a leaf of $T$)
{\sf ExtractProjection} returns a ``handle'' on the universe ``truncated''
at level $k$, that is, the universe where all nodes of level higher than $k$
are ignored (thus viewing the level $k$ nodes as leaves).
When applied to path, {\sf ExtractProjection} has a similar output.

We often say that a function (see for instance
Algorithm~\ref{Algo:Intersect}) returns a refined cylindrical decomposition.
This is another way of saying that the universe is updated
to a new state corresponding to a cylindrical decomposition
refining (in the sense of a partition of a set refining another 
partition of the same set) the cylindrical decomposition
of the previous state.

After these preliminary remarks on the pseudo-code,
we present the specifications of the algorithm
and related subroutines. 

The top level algorithm for computing a cylindrical tree
is described by Algorithm~\ref{Algo:CylindricalDecompose}. 
It takes a set $F$ of non-constant polynomials in $\K[x_1<\cdots<x_n]$
as input and returns an $F$-invariant cylindrical decomposition of $\KK^n$.
This algorithm relies on a core operation, 
called {\sf Intersect}, which computes a cylindrical decomposition 
in an incremental manner.

The {\sf Intersect} operation is described by Algorithm~\ref{Algo:Intersect}.
It takes a cylindrical tree $T$ and a polynomial $p$ of $\K[x_1<\cdots<x_n]$
as input. It refines tree $T$ such that $p$ is sign invariant above each path of 
the refined tree $T$.
This operation is achieved by refining each path of $T$ with {\sf IntersectPath}.

The {\sf IntersectPath} operation is described by Algorithm~\ref{Algo:IntersectPath}.
It takes a polynomial $p$, a  cylindrical tree $T$ and a path $\Gamma$ of $T$ in $\K[x_1<\cdots<x_n]$
as input. It refines $\Gamma$ and updates the tree $T$ accordingly such that 
$p$ is sign invariant above each path derived from $\Gamma$ in the updated tree $T$.
This operation finds the node $N$ in $\Gamma$ whose level is the same as that of $p$.
Let ${\Gamma}_N$ be the sub-path of $\Gamma$ from $N$ to the root of $T$.
The {\sf IntersectPath} operation then 
calls the routine {\sf IntersectMain} so as to refine ${\Gamma}_N$ into a tree $T_N$
such that $p$ becomes sign invariant w.r.t. $T_N$.

The routine {\sf IntersectMain} is described by Algorithm~\ref{Algo:IntersectMain}.
It takes a cylindrical tree $T$, a path $\Gamma$ of $T$, 
and a polynomial of the same level as the leaves of $T$ in $\K[x_1<\cdots<x_n]$
as input. It refines $\Gamma$ and updates the tree $T$ accordingly such that 
$p$ becomes sign invariant above each path derived from $\Gamma$ 
in the updated tree.

The  routine {\sf IntersectMain} works in the following way.
It first splits $\Gamma$ such that 
above the projection $C_{n-1}$ of each new branch $C$ of $\Gamma$ in $\KK^{n-1}$, 
the number of distinct roots of $p$ w.r.t. $x_n$ is invariant.
This is achieved by the operation {\sf Squarefree}, 
described by Algorithm~\ref{Algo:Squarefree}.
The squarefree part of $p$ above a branch $C$ 
is denoted by $sp$.
If $p$ has no roots or is identically zero above $C_{n-1}$,
the sign of $p$ above $C$ is determined immediately. 
Otherwise, a case discussion is made according to 
the structure of the leaf node $V$ of $C$.
If $V$ has no constraints associated to it, 
then $V$ is simply split into two new nodes $sp=0$ and $sp\neq 0$.
Assume now that $V$ has a constraint, which can
be either of the form $f=0$ or of the form $f\neq 0$, 
where $f$ is a level $n$ polynomial squarefree modulo $C_{n-1}$.
This case is handled by computing the GCD $g$ of $sp$ and $f$ modulo $C_{n-1}$.
The node $V$ then splits based on the GCD $g$ and the co-factors of $sp$ and $f$.

The GCD is computed by the operation ${\sf Gcd}$, 
described by Algorithm~\ref{Algo:Gcd} and~\ref{Algo:Gcdi}.
The co-factors are computed by Algorithm~\ref{Algo:CoFactor}.
The {\sf Squarefree} and {\sf Gcd} operations
rely on the operation {\sf MakeLeadingCoefficientInvertible},
described by Algorithm~\ref{Algo:InvertLeadingCoefficient}.
This latter operation takes as input a polynomial $p$ of $\K[x_1<\cdots<x_n]$,
a cylindrical tree $T$ of $\K[x_1<\cdots<x_{n-1}]$
and a path $\Gamma$ of $T$.
Then, it refines $\Gamma$ and updates $T$ accordingly such that
above each path $C$ of $T$ derived from $\Gamma$, 
the polynomial $p$ is either zero or its leading coefficient is invertible.

All the algorithms also rely on the following three operations 
which perform manipulations and traversal of the tree data structure. 
For these three operations, only specifications are provided below
while their algorithms are explained in Section~\ref{sec:datastructure}.

\begin{algorithm}
\caption{\Update{\Gamma, T}}
\label{Algo:Update}
\begin{itemizeshort}
\item [-] {Input:}
A cylindrical tree $T$.
A path $\Gamma$ in some past state of $T$. 
\item [-] {Output:} 
A subtree $ST$ in present state of $T$. 
$ST$ is derived from $\Gamma$ according to the historical data of $T$.
\end{itemizeshort}
\end{algorithm}
\begin{algorithm}
\caption{\Project{T, k}}
\label{Algo:Project}
\begin{itemizeshort}
\item [-] {Input:}
A cylindrical tree $T$ of $\K[x_1<\cdots <x_n]$.
An integer $k$, $0\leq k \leq n$.
\item [-] {Output:} 
A cylindrical tree $T_k$ in $\K[x_1<\cdots < x_k]$
such that $T_k$ is the projection of $T$ in $\K[x_1 <\cdots < x_k]$. 
\end{itemizeshort}
\end{algorithm}
\begin{algorithm}
\label{Algo:NextPathToDo}
\caption{{\sf NextPathToDo$_n$}(T)}
\begin{itemizeshort}
\item [-] {Input:}
A cylindrical tree $T$ in $\K[x_1<\cdots<x_n]$.
\item [-] {Output:} 
For a fixed traversal order of a tree, 
return the first ``ToDo'' path $\Gamma$ of $T$.
\end{itemizeshort}
\end{algorithm}

\begin{algorithm}
\label{Algo:CylindricalDecompose}
\linesnumbered
\KwIn{
$F$ is a set of non-constant polynomials in $\K[x_1<\cdots<x_n]$.
}
\KwOut{
An $F$-invariant cylindrical decomposition of $\KK^n$.
}
\caption{\CylindricalDecompose{F}}
\Begin{
create a tree $T$ with only one vertex $V_0$: the $root$ of $T$\;
\For{$i$ from $1$ to $n$}{
create a vertex $V_i$;
$V_i.signs := \emptyset$; $V_i.formula :=$ ``any $x_i$''\;
$V_{i-1}.child := V_i$; 
}
\For{$p\in F$}{
     ${\sf Intersect}_n(p, T)$\;
}
return $T$\;
}
\end{algorithm}

\begin{algorithm}
\label{Algo:Intersect}
\linesnumbered
\caption{${\sf Intersect}_n(p, T)$}
\KwIn{
A cylindrical tree $T$ of $\K[x_1<\cdots<x_n]$.
A non-constant polynomial $p$ of $\K[x_1<\cdots<x_n]$. 
}
\KwOut{
A refined cylindrical decomposition such that 
$p$ is sign invariant above each path of $T$.
}

    \While{$\Gamma := {\sf NextPathToDo}_n(T)\neq \emptyset$}{
         ${\sf IntersectPath}_n(p, \Gamma, T)$\;
    }
\end{algorithm}

\begin{algorithm}
\label{Algo:IntersectPath}
\linesnumbered
\caption{${\sf IntersectPath}_n(p, \Gamma, T)$}
\KwIn{A cylindrical tree $T$ of $\K[x_1<\cdots<x_n]$.
A path $\Gamma$ of $T$. A polynomial $p$ of $\K[x_1<\cdots<x_n]$. 
}
\KwOut{
A refined cylindrical decomposition $T$ such that
$p$ is sign invariant above each path derived from $\Gamma$.
}
\Begin{
    \uIf{$p\in\K$}{
         return;
    }
    \Else{
      $k := \level{p}$\;
      \uIf{$k=n$}{
         ${\sf IntersectMain}_n(p, \Gamma, T)$\;
      }
      \Else{
         $T_k := \Project{T, k}$;$\Gamma_k := \Project{\Gamma,k}$\;
         ${\sf IntersectMain}_k(p, \Gamma_k, T_k)$\;
         $\Update{\Gamma, T}$\;
         \For{each leaf $V$ of $\Gamma$}{
             Let $L_k$ be the ancestor of $V$ of level $k$;
             $V.signs[p] := L_k.signs[p]$ \;
         }
      }
    }
}
\end{algorithm}

\begin{algorithm}
\label{Algo:IntersectMain}
\linesnumbered
\KwIn{A cylindrical tree $T$ of $\K[x_1<\cdots<x_n]$.
A path $\Gamma$ of $T$.
A polynomial $p$ of level $n$ in $\K[x_1<\cdots<x_n]$.
}
\KwOut{
A refined cylindrical decomposition $T$ such that
$p$ is sign invariant above each path derived from $\Gamma$.
}
\caption{${\sf IntersectMain}_n(p, \Gamma, T)$}
\Begin{
    $T_{n-1} := \Project{T, n-1}$; $\Gamma_{n-1} := \Project{\Gamma, n-1}$\;
          $\Squarefree{p, \Gamma_{n-1}, T_{n-1}}$\;
    $\Update{\Gamma, T}$\;
    \While{ $C := {\sf NextPathToDo}_n(\Gamma)\neq \emptyset$}{
        $V := C.leaf$; $C_{n-1}:=\Project{C, n-1}$\;
        $sp := C_{n-1}.leaf.Squarefree[p]$\;
        \uIf{$sp=0$}{
             $V.signs[p] := 0$\;
        }
        \uElseIf{$sp=1$}{
             $V.signs[p] := 1$\;
        }
        \uElseIf{$V.formula$ is ``any $x_n$''}{
            split $V$ into two new vertices $V_1$ and $V_2$\;
            $V_1.formula := sp = 0$;
            $V_1.signs := V.signs$; $V_1.signs[p] := 0$\;
            $V_2.formula := sp \neq 0$;
            $V_2.signs := V.signs$; $V_2.signs[p] := 1$\;  
            $V_1.others := V.others$; $V_2.others := V.others$\; 
            $C_{n-1}.leaf.children := V_1,V_2$\;
        }
        \Else{
            {\small\tcp{$V.formula$ is of the form $f=0$ or $f\neq 0$}}
            $\RegularGcd{sp, f, C_{n-1}, T_{n-1}}$\;
            $\Update{C, T}$\;
            \For{each leaf $V$ of $C$}{
                let $L$ be the parent of $V$\;
                $cp, g, cf := \CoFactor{sp, L.Gcd[sp, f], f}$\;
                \uIf{$V.formula$ is of the form $f=0$}{
                     \uIf{$g=1$}{
                         $V.signs[p] := 1$\;
                     }
                     \uElseIf{$cf=1$}{
                         $V.signs[p] := 0$\;
                     }
                     \Else{
                         split $V$ into two new vertices $V_1$ and $V_2$\;
                         $V_1.formula := g = 0$;$V_1.signs := V.signs$; $V_1.signs[p] := 0$\;
                         $V_2.formula := cf = 0$;$V_2.signs := V.signs$; $V_2.signs[p] := 1$\;   
                         $V_1.others := V.others$;
                         $V_2.others := V.others$\; 
                         $L.children := V_1,V_2$\;
                     }
                }
                \Else{
                    \eIf{$cp=1$}{
                         $V.signs[p] := 1$\;
                     }{
                         split $V$ into two new vertices $V_1$ and $V_2$\;
                         $V_1.formula :=  cp = 0$; $V_1.signs := V.signs$; $V_1.signs[p] := 0$\;
                         $V_2.formula := (f*cp) \neq 0$\; $V_2.signs := V.signs$; $V_2.signs[p] := 1$\;  
                         $V_1.others := V.others$; $V_2.others := V.others$\; 
                         $L.children := V_1,V_2$\;
                    }
                }
            }
        }
      
    }
}
\end{algorithm}

\begin{algorithm}
\label{Algo:Squarefree}
\linesnumbered
\caption{\Squarefree{p, \Gamma, T}}
\KwIn{
A cylindrical tree $T$ of $\K[x_1<\cdots<x_{n-1}]$.
A path $\Gamma$ of $T$.
A polynomial $p$ of level $n$.
}
\KwOut{
A refined cylindrical tree $T$ of $\K[x_1<\cdots<x_{n-1}]$. 
Above each path $C$ of $T$ derived from $\Gamma$, 
there is a dictionary $C.leaf.Squarefree$.  
Let $p^* := C.leaf.Squarefree[p]$. We have:
\begin{itemizeshort}
\item $p=p^*$ modulo $C$.
\item If $p^*$ is of level $n$, then both $\init{p^*}$ and $\discrim{p^*}$
are invertible modulo  $C$.
\item If $p^*$ is of level less than $n$, then $p^*$ is 
either $0$ or $1$.
\end{itemizeshort}
}
\Begin{
    \If{$n=1$}{
        let $r$ be the root of $T$;
        $r.Squarefree[p] := \SquarefreePart{p}$\;
        return
    }
    \RegularizeInitial{p, p, \Gamma, T}\;
    \While{$C :=  \NextPathToDo{\Gamma}\neq \emptyset$}{
        $f := C.leaf.InvertLc[p]$\;
        \uIf{$\level{f}<n$ or $\deg(f,x_n)=1$}{
            $C.leaf.Squarefree[p] := f$
        }
        \Else{
            \RegularGcd{f,\der{f}, C, T}\;
            \For{each leaf $L$ of $C$}{
                $g := L.Gcd[f, \der{f}]$\;
                \uIf{$g=1$}{
                     $L.Squarefree[p] := f$
                }
                \Else{
                     $L.Squarefree[p] := \pquo{f, g}$ 
                }
            }
        }
    }
}
\end{algorithm}

\begin{algorithm}
\linesnumbered
\caption{\RegularGcd{p, f, \Gamma, T}}
\label{Algo:Gcd}
\KwIn{
 A cylindrical tree $T$ of $\K[x_1<\cdots<x_{n-1}]$.
A polynomial $p\in \K[x_1<\cdots<x_n]$ of level $n$.
 A path $\Gamma$ of $T$.
A polynomial $f$ of level $n$ such that $\init{f}$ 
is invertible modulo $\Gamma$.
}
\KwOut{
A refined cylindrical tree $T$. 
Above each path $C$ of $T$ derived from $\Gamma$, 
there is a dictionary $C.leaf.Gcd$ 
such that $C.leaf.Gcd[p, f]$ is a GCD of $p$ and $f$ modulo $C$. 
}
\Begin{
    let $S$ be the subresultant chain of $p$ and $f$\;
    \eIf{$\mdeg{p}\geq \mdeg{f}$}{
        $d := \mdeg{f}$
    }{
        $d := \mdeg{p}+1$  
    }
    return $\RegularGcd{p, f, S, d, 0, \Gamma, T}$\;
}
\end{algorithm}

\begin{algorithm}
\label{Algo:Gcdi}
\linesnumbered
\caption{\RegularGcd{p, f, S, d, i, \Gamma, T}}
\KwIn{
\begin{itemizeshort}
\item A polynomial $p\in \K[x_1<\cdots<x_n]$ of level $n$.
\item A polynomial $f$ of level $n$ such that $\lc{f}$ is invertible modulo $\Gamma$.
\item The subresultant chain $S$ of $p$ and $f$ w.r.t. $x_n$.
\item A non-negative integer $d$ (as defined in the pseudo-code of 
      Algorithm~\ref{Algo:Gcd}) and such that the principle subresultant coefficient 
      $s_d$ is invertible modulo $\Gamma$.
\item A non-negative integer $i$ such that $0 \leq i \leq d$ and 
      the principle subresultant coefficient $s_j$ is zero 
      modulo $\Gamma$, for all $0\leq j < i$.
\item A path $\Gamma$ of $T$.
\item A cylindrical tree $T$ of $\K[x_1<\cdots<x_{n-1}]$.
\end{itemizeshort}
}
\KwOut{
A refined cylindrical tree $T$. 
Above each path $C$ of $T$ derived from $\Gamma$, 
there is a dictionary $C.leaf.Gcd$ 
such that $C.leaf.Gcd[p, f]$ is a GCD of $p$ and $f$ modulo $C$. 
}
\Begin{
    \If{$i=d$}{
        $\Gamma.leaf.Gcd[p, f]:=S_i$\;
        return\;
    }
    ${\sf IntersectPath}_{n-1}(s_i, \Gamma, T)$\;
    \While{$C :=  \NextPathToDo{\Gamma}\neq\emptyset$}{
        \uIf{$C.leaf.signs[s_i]=1$}{
             \uIf{$i=0$}{
                 $C.leaf.Gcd[p, f] := 1$
             }
             \Else{
                 $C.leaf.Gcd[p, f] := S_i$ 
             }
        }
        \Else{
             
             $\RegularGcd{p, f, S, d, i+1, C, T}$    
        }
    }
}
\end{algorithm}

\begin{algorithm}
\label{Algo:CoFactor}
\linesnumbered
\caption{\CoFactor{p, g, f}}
\KwIn{
Two polynomials $p$ and $f$ of level $n$ in $\K[x_1<\cdots<x_n]$.
A polynomial $g$ which is either $1$ or of level $n$ in $\K[x_1<\cdots<x_n]$.
}
\KwOut{
As described by the algorithm.
}
\Begin{
                     \uIf{$g=1$}{
                          $cp := p$; $gg := 1$;  $cf := f$\;
                     }
                     \uElseIf{$\mdeg{g}=\mdeg{f}$}{
                          $gg := f$\;
                          \eIf{$\mdeg{g}=\mdeg{p}$}{
                               $cf := 1$; $cp := 1$\;
                          }{
                               $cf := 1$; $cp := \pquo{p, gg}$
                          }
                     }
                     \uElseIf{$\mdeg{g}=\mdeg{p}$}{
                          $gg := p$; $cf := \pquo{f,gg}$; $cp := 1$\; 
                     }
                     \Else{
                           $cp := \pquo{p, g}$; $cf := \pquo{f,g}$; $gg := g$\; 
                     }
return $cp, gg, cf$\;
}
\end{algorithm}

\begin{algorithm}
\label{Algo:InvertLeadingCoefficient}
\linesnumbered
\caption{\RegularizeInitial{p, \bar{p}, \Gamma, T}}
\KwIn{
A polynomial $p$ of $\K[x_1<\cdots<x_n]$.
A polynomial $\bar{p}$ of $\K[x_1<\cdots<x_n]$ 
such that $p=\bar{p}$ modulo $\Gamma$.
A cylindrical tree $T$ of $\K[x_1<\cdots<x_{n-1}]$.
A path $\Gamma$ of $T$.
}
\KwOut{
A refined cylindrical tree $T$ of $\K[x_1<\cdots<x_{n-1}]$. 
Above each path $C$ of $T$ derived from $\Gamma$, 
there is a dictionary $C.leaf.InvertLc$. 
Let $p^*$ be the polynomial $C.leaf.InvertLc[p]$.
Then, we have:
\begin{itemizeshort}
\item $p=p^*$ modulo $C$.
\item If $p^*$ is of level $n$, then $\init{p^*}$
is invertible modulo the path $C$.
\item If $p^*$ is of level less than $n$, then $p^*$ is 
either $0$ or $1$.
\end{itemizeshort}
}
\Begin{
    ${\sf IntersectPath}_{n-1}(\lc{\bar{p}, x_n}, \Gamma, T)$\;
    \While{$C :=  \NextPathToDo{\Gamma}\neq\emptyset$}{
        \uIf{$C.leaf.signs[\lc{\bar{p},x_n}]=1$}{
             \uIf{$\level{\bar{p}}<n$}{
                  $C.leaf.InvertLc[p] := 1$
             }
             \Else{
                  $C.leaf.InvertLc[p] := \bar{p}$
             }
        }
        \Else{
             \uIf{$\level{\bar{p}}<n$}{
                  $C.leaf.InvertLc[p] := 0$
             }
             \Else{
                  \RegularizeInitial{p, \tail{\bar{p}}, C, T}
             }
        }
    }
}
\end{algorithm}

\begin{Theorem}
For a set of polynomials in $\K[x_1,\ldots,x_n]$, 
Algorithm~\ref{Algo:CylindricalDecompose} computes 
an $F$-invariant cylindrical decomposition of $\KK^n$.
\end{Theorem}
\begin{proof}
Firstly, we prove the termination.
The basic mutual calling graph of its subroutines are:
$$
{\sf IntersectMain}_n\rightarrow{\sf Squarefree_n}
\rightarrow{\sf IntersectMain}_{n-1}\rightarrow\cdots,
$$
and 
$$
{\sf IntersectMain}_n\rightarrow{\sf Gcd_n}
\rightarrow{\sf IntersectMain}_{n-1}\rightarrow\cdots
$$
So the termination is easily proved by induction.
The correctness follows from the specification 
of its subroutines and Theorem~\ref{Theorem:gcd}.
\end{proof}

\begin{Example}
In this example, we illustrate the operation {\sf IntersectPath}. 
Let $F := \{y^2+x, y^2+y\}$. 
The incremental algorithm first computes an $y^2+x$
sign invariant complex cylindrical tree, 
which is described by the following tree $T$.
$$
T := \left\{
\begin{array}{ll}
{x = 0}  & \left\{
        \begin{array}{rcl}
         {y = 0}   &:& { y^2+x=0}\\
         y\neq 0 &:& { y^2+x\neq0}
        \end{array}
        \right.\\
&\\
x\neq 0& \left\{
        \begin{array}{rcl}
         y^2+x = 0   &:& { y^2+x=0} \\
         y^2+x\neq 0 &:& { y^2+x\neq0}
        \end{array}
        \right.\\ 
\end{array}
\right.
$$
Let $\Gamma$ be the path $\{x=0, y\neq 0\}$ of $T$. 
Calling $\IntersectPath{y^2+y, \Gamma, T}$
will update $T$ into the following tree.

$$
\left\{
\begin{array}{ll}
{ x = 0}  & \left\{
        \begin{array}{rcl}
         {y = 0}   &:& { y^2+x=0}\\
         { y = -1}  &:& { y^2+x\neq0\wedge y^2+y=0}\\
         { \rm otherwise} &:& { y^2+x\neq0\wedge y^2+y\neq 0}\\
        \end{array}
        \right.\\
&\\
x\neq 0& \left\{
        \begin{array}{rcl}
         y^2+x = 0   &:& { y^2+x=0} \\
         y^2+x\neq 0 &:& { y^2+x\neq0}
        \end{array}
        \right.\\ 
\end{array}
\right.
$$

\end{Example}

\section{Building a CAD tree from a complex cylindrical tree}
\label{sec:cad}
In this section, we review briefly how to compute a
CAD of $\R^n$ from a cylindrical decomposition of $\C^n$.
The reader may refer to~\cite{CMXY09} for more details.
Recall that $n \geq 1$ holds. We denote by ${\pi}_{n-1}$  the
standard projection from ${\R}^n$ to ${\R}^{n-1}$ that
maps $(x_1, \ldots, x_{n-1}, x_n)$ onto $(x_1, \ldots, x_{n-1})$.

\smallskip\noindent{\small \bf Stack over a connected semi-algebraic set.}
Let $S$ be a connected semi-algebraic subset of $\R^{n-1}$. 
The {\em cylinder} over $S$ in $\R^n$ is defined as $Z_{\R}(S) := S\times\R$.
Let $\theta_1<\cdots<\theta_s$ be continuous semi-algebraic functions 
defined on $S$. The intersection of the graph of $\theta_i$ with $Z_{\R}(S)$
is called the {\em $\theta_i$-section} of $Z_{\R}(S)$.
The set of points between two consecutive sections of $Z_{\R}(S)$
is a connected semi-algebraic subset of $\R^n$, 
called a {\em sector} of $Z_{\R}(S)$.
All the sections and sectors of $Z_{\R}(S)$ form a disjoint 
decomposition of $Z_{\R}(S)$, called a {\em stack} over $S$.

\smallskip\noindent{\small \bf Cylindrical algebraic decomposition.}
A finite partition ${\DD}$ of $\R^n$ is 
called a {\em cylindrical algebraic decomposition} (CAD) of $\R^n$ if 
one of the following properties holds.
\begin{itemizeshort}
\item Either $n=1$ and ${\DD}$ is a stack over $\R^0$.
\item Or the set of $\{ {\pi}_{n-1} (D) \ | \ D \in {\DD} \}$
      is  a CAD of $\R^{n-1}$ and each $\ D \in {\DD}$ is a section 
      or sector of the stack over ${\pi}_{n-1} (D)$.
\end{itemizeshort}  
When this holds, the  elements of ${\DD}$ are called {\em cells}.

\smallskip\noindent{\small \bf Sign invariance and delineability.}
Let $p$ be a polynomial of $\R[x_1,\ldots,x_n]$, and let $S$ be a subset of $\R^n$.
The polynomial $p$ is  called {\em sign invariant} on  $S$
if the sign of $p(\alpha)$ does not change when $\alpha$ ranges over $S$. 
Let $F\subset\R[x_1,\ldots,x_n]$ be
a finite polynomial set. We say $S$ is $F$-invariant if each $p\in F$ is
invariant on $S$. A cylindrical algebraic decomposition ${\DD}$ is $F$-invariant if 
$F$ is invariant on each cell $D \in {\DD}$.
Let $p$ be a polynomial of $\R[x_1,\ldots,x_n]$, 
and let $S$ be a connected semi-algebraic set of $\R^{n-1}$. We say that $p$ 
is {\em delineable} on $S$ if the real zeros of $p$
define continuous semi-algebraic functions 
$\theta_1,\ldots,\theta_s$ such that, for all $\alpha\in S$ we have 
$\theta_1(\alpha)<\cdots<\theta_s(\alpha)$.
In other words, $p$  is delineable on $S$ if its real zeros
naturally determine a stack over $S$. 
We recall the following Theorem introduced in~\cite{CMXY09}.
\begin{Theorem}
\label{Theorem:stack}
Let $P=\{p_1,\ldots,p_r\}$ be a finite set of polynomials 
in $\R[x_1<\cdots<x_n]$ of level $n$. 
Let $S$ be a connected semi-algebraic subset of $\R^{n-1}$. 
If $P$ {\em separates} above $S$, 
then each $p_i$ is delineable on $S$.
Moreover, the product of the $p_1,\ldots,p_r$ is also delineable on $S$.
\end{Theorem}

Let $F$ be a finite set of polynomials in $\Q[x_1<\cdots<x_n]$. 
Let $CT$ be an $F$-invariant complete cylindrical tree of $\C^n$.
Applying Theorem~\ref{Theorem:stack} to polynomials in $CT$, we can derive 
an $F$-invariant cylindrical algebraic decomposition of  $\R^n$ by induction on $n$.
A procedure {\sf MakeSemiAlgebraic}, was introduced in~\cite{CMXY09}
to derive a CAD from a $CT$ via real root isolation of zero-dimensional  regular chains.
\smallskip

\begin{Example}
Let $F := \{y^2+x\}$. 
An $F$-invariant cylindrical algebraic decomposition 
is described by the following tree.
$$
T := \left\{
\begin{array}{ll}
x < 0  & \left\{
        \begin{array}{lcl}
         y < -\sqrt{|x|}   &:& { y^2+x>0}\\
         y = -\sqrt{|x|}   &:& { y^2+x=0}\\
         y > -\sqrt{|x|}\wedge y<\sqrt{|x|}   &:& { y^2+x<0}\\
         y = \sqrt{|x|}   &:& { y^2+x=0}\\
         y > \sqrt{|x|}   &:& { y^2+x>0}
        \end{array}
        \right.\\
&\\
x = 0  & \left\{
        \begin{array}{rcl}
         y <0    &:& { y^2+x>0}\\
         y = 0   &:& { y^2+x=0}\\
         y > 0 &:& { y^2+x>0}
        \end{array}
        \right.\\
&\\
x> 0&   \;{\rm for~ any~}y\,  :\;\, { y^2+x > 0}
\end{array}
\right.
$$
\end{Example}

\section{Making use of equational constraints and other optimizations}
\label{sec:equation}
In this section, we discuss several possible optimizations 
to algorithms presented in Section~\ref{sec:incremental}. 

Firstly, we discuss how to compute a CAD dedicated to a semi-algebraic system, 
which provides a systematic solution for making use of equational constraints when computing CADs.
The motivation for making use of equational constraints
comes from quantifier elimination. 
Let $$PF := (Q_{k+1}x_{k+1}\cdots Q_nx_n)FF(x_1,\ldots,x_n),$$ be 
a prenex formula, where $FF$ is a DNF formula. 
To perform {\QE} by {\CAD},
 the first computation step is to collect all the polynomials 
appearing in $FF$ as a polynomial set $F$ and compute an $F$-invariant CAD of $\R^n$.
This process of computing an $F$-invariant CAD  exhausts all possible sign combinations of $F$, 
including those which do not appear in $FF$, 
and thus often computes much more than needed for solving the input QE problem.
Different techniques in the literature have been proposed for taking advantage of the structure of the input problem.
These methods include partial CAD~\cite{ch91} for lazy lifting, 
simplified projection operator for handling pure strict inequalities~\cite{scott93,adam00}, 
smaller projection sets for making use of equational constraints~\cite{Collins98, McCallum2001, Brown05, McCallum2009}.

To make the discussion clear, we first quote a paragraph of~\cite{Brown05}.
``The idea is as follows: if an input formula includes the constraint $f=0$, 
then decompose $\R^r$ into regions in which $f$ has invariant sign, 
and then refine the decomposition so that the other polynomials 
have invariant sign in those cells in which $f=0$. 
The signs of the other polynomials in cells in which $f\neq0$ are, 
after all, irrelevant. 
Additionally, the method of equational constraints seeks 
to deduce and use constraints that are not explicit in the input 
formula, but rather arise as consequences of two or more explicit
constraints (e.g. if $f=0$ and $g=0$ are explicit constraints, 
then $\res{f,g}=0$ is also a constraint.)''

This idea, of course, is attractive.
Much progress on it has also been made. 
However, the reason why it is a generally hard problem 
for CAD is that the framework of {\CCAD} does not have 
much flexibility to allow propagation of equational constraints. 
In the world of {\CCAD}, one always tries to obtain a generic 
projection operator and then applies the same projection operator 
recursively. To obtain a generic projection operator 
for handling equational constraints is hard
because many problems inherently require different projection operators
during projection. 
Therefore case discussion is important.

In fact, case discussion is very common in algorithms for computing 
triangular decompositions.
For such algorithms, equational constraints are natural input of these algorithms. 
The two keys ideas ``splitting only above $f=0$'' and  ``if $f=0$ and $g=0$ are explicit constraints, 
then $\res{f,g}=0$ is also a constraint'' have already been systematically 
taken care of in the {\sf Intersect} operation of the authors' paper
for  computing triangular decompositions~\cite{CM11}.

Next we explain how to modify algorithms presented in Section~\ref{sec:incremental} to 
automatically implement these ideas.

Suppose now that the input of Algorithm {\sf CylindricalDecompose} is a 
system of equations or inequations, 
this algorithm will then compute a partial cylindrical tree 
such that its zero set is exactly the zero set of input system.
This can be simply achieved by passing an equation or inequation 
to the function {\sf Intersect}.
W.l.o.g., let us assume that an equation $p=0$ is passed as 
an argument of {\sf Intersect}.
Then for this function and all its called subroutines, 
we will cut the computation branches above which $p$ is known to be nonzero 
and never proceed with computation branches above which $p$ cannot be zero.
For example, we will not create a new vertex at step $15, 32, 42$ 
in Algorithm {\sf IntersectMain}.
We will delete the vertex $V$ at step $11$, $26$, $37$ since $p$ is nonzero on $V$.

The first important optimization in {\sf IntersectMain} which can be implemented 
is to avoid {\sf Squarefree} computation at step $3$ if $\Gamma.leaf$
is an equational constraint.
This idea is quite close to ``splitting only above $f=0$''.
Another important optimization can be done 
at step $19$ of {\sf IntersectMain}. 
Assume that $V.formula$ is an equational constraint $f=0$, 
then when {\sf Gcd} is called, 
in step $5$ of Algorithm~\ref{Algo:Gcdi},
we can do as follows. 
If $i=0$, then $s_i$ is the resultant of $p$ and $f$.
Thus we should pass $s_i=0$ to the {\sf IntersectPath} operation
in order to avoid useless computations 
on the branch $s_i\neq 0$.
This addresses the idea ``if $f=0$ and $g=0$ are explicit constraints, 
then $\res{f,g}=0$ is also a constraint.''
Moreover, these optimizations are systematically performed during the whole 
computation.

Next we briefly mention several other important optimizations.
Let $V$ be a leaf of a path $\Gamma$ of a cylindrical tree. 
Assume that $V.formula$ is of the form $f\neq 0$ or of  the form  $f=0$.
We can safely replace $f$ by its primitive part since $\lc{f}$
is invertible modulo $\Gamma_{n-1}$. 
Replacing $f$ by its irreducible factors over $\Q$ is often a 
more efficient choice.
Last but not least,
recall that a path $\Gamma$ in the cylindrical tree is a simple system.
Writing $\Gamma$ as two parts $\Gamma := [T, H]$, where $T$ is 
a set of equations and $H$ is a set of inequations. 
We know that $T$ is a regular chain and $\Gamma$ is a squarefree regular system. 
Thus the Zariski closure of $\Gamma$ is the variety of the saturated ideal of $T$.
We can call the pseudo division operation $\prem{p, T}$ or $\prem{f, T}$
to test whether $p$ or $f$ is zero modulo $\Gamma$. 
And sometimes replacing $p$ by $\prem{p, T}$ and $f$ by $\prem{f, T}$
also ease the computations.

\begin{Example}
Let $F:= \{y^2+x=0, y^2+y=0\}$ be a system of equations. 
Taking $F$ as input, Algorithm {\sf CylindricalDecompose} generates the 
following partial cylindrical tree $T$ of $\C^2$ such that the zero set of $F$
is exactly the union of the zero sets of the paths in $T$.
\begin{figure}[htbp]
\begin{center}
\scalebox{0.4}{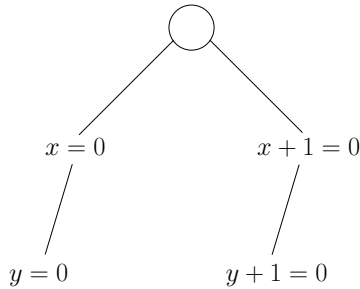 }
\caption{A partial cylindrical tree $T$ adapted to $F$}
\label{fig:tree-eqs}
\end{center}
\end{figure}
\end{Example}

\section{Benchmark}
\label{sec:benchmark}
In this section, we report on the experimental results 
of a preliminary implementation 
in the {\RegularChains} 
library of {\Maple} of the algorithms of Sections~\ref{sec:incremental}
and \ref{sec:cad}.

The examples in Table~\ref{table:cd} and Table~\ref{table:cad}
are from papers on polynomial system solving,
such as~\cite{CGLMP07,BoulierChenLemaireMorenoMaza09}
and the references therein.
All the tests were launched
on a machine with Intel Core 2 Quad {\small CPU} (2.40{\small GHz}) 
and 8.0{\small Gb} total memory.
The time-out is set as $1$ hour.
In the tables, the symbol $> 1h$ means time-out.

The {\Maple} functions are launched in {\Maple}~15
with the latest {\RegularChains} library.
The memory usage is limited to $60\%$ of total memory.
The software {\QEPCAD} is launched with the option $+N500000000 +L200000$, 
where the first option specifies the memory to be pre-allocated 
(about $23\%$ of total memory for our machine)
and the second option specifies the number of prime numbers to be used.

\begin{table}
\centering
\caption{Timings for computing cylindrical decomposition of the complex space}
\label{table:cd}
\begin{tabular}{|c|c|c|c||c|c|c|c|}\hline
System             & tcd-rec & tcd-inc & tcd-eqs & System & tcd-rec & tcd-inc & tcd-eqs \\ \hline
AlkashiSinus & 	  3373.966 & 	  14.568 & 	  4.168 &MontesS10 &  	 $> 1h$ &  	 $> 1h$ & 	  2.952 \\
Alonso & 	  9.636 & 	  1.404 & 	  0.700 &MontesS12 &  	 $> 1h$ &  	 $> 1h$ & 	  7.528 \\
Arnborg-Lazard-rev & 	  2759.940 & 	  2419.543 & 	  16.233 &MontesS15 &  	 $> 1h$ &  	 $> 1h$ & 	  77.048 \\
Barry & 	  39.346 & 	  1.808 & 	  0.556 &MontesS16 &  	 $> 1h$ &  	 $> 1h$ & 	  8.228 \\
blood-coagulation-2 & 	  235.310 & 	  9.472 & 	  0.808 &MontesS4 & 	  556.390 & 	  102.122 & 	  0.488 \\
Bronstein-Wang & 	  255.427 & 	  35.990 & 	  1.120 &MontesS5 & 	  1449.810 & 	  119.059 & 	  1.004 \\
cdc2-cyclin &  	 $> 1h$ & 	  68.920 & 	  65.976 &MontesS7 &  	 $> 1h$ &  	 $> 1h$ & 	  1.060 \\
circles & 	  276.389 & 	  2.280 & 	  0.520 &MontesS9 & 	  269.636 & 	  4.212 & 	  0.980 \\
genLinSyst-3-2 & 	  916.245 & 	  19.537 & 	  1.384 &nql-5-4 &  	 $> 1h$ & 	  1.056 & 	  0.528 \\
genLinSyst-3-3 &  	 $> 1h$ & 	  160.406 & 	  12.408 &r-5 & 	  68.364 & 	  3.232 & 	  0.876 \\
Gerdt &  	 $> 1h$ &  	 $> 1h$ & 	  1.188 &r-6 & 	  1456.883 & 	  46.458 & 	  1.200 \\
GonzalezGonzalez & 	  141.072 & 	  53.451 & 	  0.732 &Raksanyi & 	  1471.351 & 	  118.227 & 	  1.000 \\
hereman-2 &  	 $> 1h$ & 	  40.042 & 	  0.908 &Rose &  	 $> 1h$ & 	  51.855 & 	  1.072 \\
lhlp5 & 	  31.069 & 	  3.984 & 	  0.648 &Wang93 &  	 $> 1h$ &  	 $> 1h$ & 	  18.877 \\
Maclane &  	 $> 1h$ &  	 $> 1h$ & 	  6.420 &YangBaxterRosso & 	  54.895 & 	  1.560 & 	  0.844 \\
\hline
\end{tabular}
\end{table}

\begin{table}
\centering
\caption{Timings for computing {\CAD}}
\label{table:cad}
\begin{tabular}{|c|c|c|c|c|c|}\hline
System & qepcad & qepcad-eqs & mathematica-eqs & tcad & tcad-eqs\\ \hline
Alonso & 7.516 & 5.284 & 0.74& 	  61.591 & 	  5.776 \\
Arnborg-Lazard-rev &$> 1h$ &$> 1h$ & 0.952&  	 $> 1h$ & 	  17.325 \\
Barry & Fail& 216.425 & 0.032& 	  8.580 & 	  1.004 \\
blood-coagulation-2 & $> 1h$ & $> 1h$ & $> 1h$& 	  985.709 & 	  7.260 \\
Bronstein-Wang & $> 1h$ & $> 1h$ & 26.726& 	  333.892 & 	  2.564 \\
cdc2-cyclin & $> 1h$ & $> 1h$ & 0.208& 	  574.127 & 	  503.863 \\
circles & 21.633 & 5.996 & 41.211&  	 $> 1h$ & 	  40.902 \\
GonzalezGonzalez & 10.528 & 10.412 & 0.012& 	  214.213 & 	  1.136 \\
lhlp2 & 960.756 & 5.076 & 0.016& 	  3.124 & 	  0.952 \\
lhlp5 & 10.300 & 10.068 & 0.016& 	  35.338 & 	  1.084 \\
MontesS4 & $> 1h$  & $> 1h$ & 0.004& 	  2682.391 & 	  0.888 \\
MontesS5 & Fail & Fail & $> 1h$&  	 $> 1h$ & 	  9.400 \\
nql-5-4 & 93.073 & 5.420 & 1303.07& 	  113.675 & 	  1.004 \\
r-5 &$> 1h$ & 1802.676 & 0.016& 	  1282.928 & 	  1.208 \\
r-6 &$> 1h$ & $> 1h$ & 0.024&  	 $> 1h$ & 	  1.500 \\
Rose & Fail& $> 1h$  & $> 1h$& 	  606.361 & 	  3.136 \\
AlkashiSinus &      $> 1h$               &      $> 1h$               & 2.232&  	 $> 1h$ & 	  58.775 \\
genLinSyst-3-2 &   Fail     &   Fail    & 217.062& 	  3013.764 & 	  6.588 \\
MontesS10 &  $> 1h$  &  $> 1h$  & $> 1h$&  	 $> 1h$ & 	  22.797 \\
MontesS12 &  $> 1h$ & $> 1h$  & $> 1h$&  	 $> 1h$ & 	  330.996 \\
MontesS15 &$> 1h$ &$> 1h$ & 0.004&  	 $> 1h$ & 	  395.964 \\
MontesS7 & $> 1h$ &$> 1h$ & 245.807&  	 $> 1h$ & 	  2.452 \\
MontesS9 &Fail & Fail & $> 1h$& 	  110.902 & 	  4.944 \\
Wang93 & Fail& Fail& $> 1h$&  	 $> 1h$ & 	  152.673 \\

\hline
\end{tabular}
\end{table}

In Table~\ref{table:cd}, we report on timings 
for computing cylindrical decomposition of the complex space
with different algorithms and options.
Each input system is a set of polynomials.
The notation tcd-rec denotes an 
implementation of the original recursive algorithm in~\cite{CMXY09},
while the notation tcd-inc denotes the incremental algorithm
presented in Section~\ref{sec:incremental}.
Both  tcd-rec and tcd-inc take a set of polynomials as input. 
The notation tcd-eqs refers to an optimized version of 
tcd-inc which makes use of equational constraints, as 
explained in Section~\ref{sec:equation}.
With the implementation tcd-eqs, every input polynomial set
is regarded as a set of equations (equating each input polynomial
to zero).
As we can see in Table~\ref{table:cd}, the incremental algorithm presented 
in this paper is much more efficient than the original recursive algorithm. 
The timings of tcd-eqs show that the optimizations
presented in Section~\ref{sec:equation}
for making use of equational constraints are very effective.

In Table~\ref{table:cad}, we report on timings 
for computing {\CAD} with three different computer algebra
packages: {\QEPCAD}, the {\sf CylindricalDecomposition} command 
of {\sf Mathematica} and the algorithm
presented in Section~\ref{sec:incremental}.
Each system is a set of polynomials.
Two categories of experimentation are conducted.
The first category is concerned with the timings 
for computing a full {\CAD} of a set of polynomials.
For {\sf Mathematica}, we cannot find
any options of {\sf CylindricalDecomposition} 
for computing a full {\CAD} of a set of 
polynomials.
Therefore for this category, 
only the timings of {\QEPCAD} and {\TCAD} are reported.
The second category is concerned with 
the timings for computing a {\CAD} of a variety.
For this category, the timings for {\QEPCAD}, {\sf Mathematica}
and {\sf TCAD} are all reported.

The notation qepcad denotes computations 
that {\QEPCAD} performs by (1) treating each input system as a set of 
non-strict inequalities and, (2) treating all variables as free 
variables and, (3) executing with the ``full-cad'' option.
The notation tcad corresponds to computations that {\TCAD}
performs by (1) treating each input system as a set of 
non-strict inequalities and,
(2) computing a sign invariant full {\CAD} of 
polynomials in the input system and, 
(3) selecting the cells which satisfy those non-strict inequalities. 
In this way, both qepcad and {\TCAD} compute a full {\CAD} of a set of 
polynomials.

The notation qepcad-eqs denotes the computations 
that {\QEPCAD} performs by (1) treating each input system as a set of 
equations and, (2) treating all variables as free 
variables and, (3) executing with the default option.
The notation mathematica-eqs represents
computations where the {\sf CylindricalDecomposition} command 
of {\sf Mathematica} treats each input system as a set of equations.
The notation tcad-eqs corresponds to  computations where {\TCAD}
treats each input system as a set of equations.

From Table~\ref{table:cad}, we make the following observations.
When full {\CAD}s are computed, within one hour time limit, 
{\QEPCAD} only succeeds on $6$ out of $24$ examples while
{\TCAD} succeeds on $14$ out of $24$ examples. 
When {\CAD}s of varieties are computed, 
for all the $10$ out of $24$ examples that {\QEPCAD}
can solve within one hour time limit, 
both {\sf Mathematica} and {\sf TCAD}
succeed with usually less time. 
For the rest $14$ examples, {\sf TCAD}
solves all of them while {\sf Mathematica} 
only succeeds on $7$ of them.

\section{Conclusion}
In this paper, we present an incremental algorithm for computing CADs.
A key part of the algorithm is an {\sf Intersect} 
operation for refining a given complex cylindrical tree. 
If this operation is supplied with an equational constraint,
it only computes a partial cylindrical tree, which 
provides an automatic solution for propagating equational constraints.
We have implemented our algorithm in {\sc Maple}. 
The experimentation 
shows that the new algorithm is much more efficient than
our previous recursive algorithm. 
We also compared our implementation with the software packages 
{\QEPCAD} and {\sf Mathematica}. 
For many examples, our implementation outperforms the other two. 
This incremental algorithm can support quantifier elimination. 
We will present this work in a future paper.

\section*{Acknowledgments} 
The authors would like to thank the
readers who helped improve the earlier versions 
of this paper.
This research was supported by 
Academic Development Fund ADF-Major-27145
of The University of Western Ontario.

\bibliographystyle{plain}

\end{document}